\documentclass[11pt,letterpaper]{article}

\usepackage{url}
\usepackage{hyperref}
\usepackage[pdftex]{graphicx}
\usepackage{threeparttable}
\usepackage[margin=1in]{geometry}
\usepackage{amsmath,amssymb,amsfonts,setspace}			
\usepackage{amsthm}
\usepackage{latexsym}
\usepackage{newtxtext,newtxmath}       
\usepackage{comment}
\usepackage{footnote}
\usepackage{bm}
\usepackage[normalem]{ulem}
\newcommand{\Izumi}[1]{}
\newcommand{\Izurep}[2]{}
\newcommand{\Kitamura}[1]{}
\newcommand{\Izumig}[1]{}
\theoremstyle{plain}
\theoremstyle{mystyle}
\newtheorem{theorem}{Theorem}
\newtheorem{lemma}{Lemma}

\newtheorem{definition}{Definition}
\newcommand{\Dist}{\mathrm{dist}}

\begin{document}
\title{Low-Congestion Shortcut and Graph Parameters\\\vspace{3mm}}
\author{Naoki Kitamura\footnotemark[1] \and Hirotaka Kitagawa\footnotemark[1] \and Yota Otachi\footnotemark[2] \and Taisuke Izumi\footnotemark[1]\footnote{Nagoya Institute of Technology, Japan}\footnote{Kumamoto University, Japan}}
\date{}
\maketitle
\thispagestyle{empty}
\setcounter{page}{1}

\maketitle
\begin{abstract} 
Distributed graph algorithms in the standard CONGEST model often exhibit
the time-complexity lower bound of $\tilde{\Omega}(\sqrt{n} + D)$ rounds 
for many global problems, where $n$ is the number of nodes and
$D$ is the diameter of the input graph. Since such a lower bound is derived 
from special ``hard-core'' instances, it does not necessarily apply to 
specific popular graph classes such as planar graphs. The concept of 
\emph{low-congestion shortcuts} is initiated by Ghaffari and Haeupler [SODA2016]
for addressing the design of CONGEST algorithms running fast in restricted 
network topologies. Specifically, given a specific graph class $X$, 
an $f$-round algorithm of constructing shortcuts of quality $q$ for 
any instance in $X$ results in $\tilde{O}(q + f)$-round algorithms of solving several fundamental graph problems such as minimum spanning tree and minimum cut, 
for $X$. The main interest on this line is to identify the graph classes 
allowing the shortcuts which are efficient in the sense of breaking 
$\tilde{O}(\sqrt{n}+D)$-round general lower bounds.

In this paper, we consider the relationship between the quality of 
low-congestion shortcuts and three major graph parameters, chordality, 
diameter, and clique-width. The main contribution of the paper is threefold:
(1) We show an $O(1)$-round algorithm which constructs a low-congestion 
shortcut with quality $O(kD)$ for any $k$-chordal graph, and prove that the quality and running time of this construction is nearly optimal up to polylogarithmic factors. (2) We present 
two algorithms, each of which constructs a low-congestion shortcut with 
quality $\tilde{O}(n^{1/4})$ in $\tilde{O}(n^{1/4})$ rounds
for graphs of $D=3$, and that with quality $\tilde{O}(n^{1/3})$ in 
$\tilde{O}(n^{1/3})$ rounds for graphs of $D=4$ respectively. 
These results obviously deduce two MST algorithms running 
in $\tilde{O}(n^{1/4})$ and $\tilde{O}(n^{1/3})$ rounds for $D=3$ and $4$
respectively, which almost close the long-standing complexity gap of the MST construction in small-diameter graphs originally posed by Lotker et al. 
[Distributed Computing 2006]. 
(3) We show that bounding clique-width does not help the construction of good shortcuts by presenting a network topology of clique-width six where the construction of MST is as expensive as the general case.
\end{abstract}

\section{Introduction}
\subsection{Background}
The \emph{CONGEST} is one of the standard message-passing models in
the development of distributed graph algorithms, especially for global problems
such as shortest paths and minimum spanning tree. It is a round-based synchronous
system where each link can transfer $O(\log n)$-bit information per one 
round ($n$ is the number of nodes in the system). Since most of global distributed 
tasks as mentioned above inherently require each node to access the information far apart from itself, it is not
possible to ``localize'' the communication assessed for solving those tasks.
That is, the $\Omega(D)$-round complexity often becomes 
an \emph{universal} lower bound applied to any network topology, 
where $D$ is the diameter of the input
topology. While $D$-round computation is sufficiently long to make some information reach all the nodes in the network, the constraint of limited 
bandwidth precludes the centralized solution that one node collects the information of whole network topology because it results in expensive
$\Omega(n)$-round time complexity. The round complexity of CONGEST algorithms
solving global tasks is typically represented in the form of $\tilde{O}(n^c + D)$ 
or $\tilde{O}(n^cD)$ for some constant $0 \leq c \leq 2$\footnote{$\tilde{O}(\cdot)$ is a notation which ignores $\mathrm{polylog}(n)$ factors from $O(\cdot)$.}, 
and thus the main complexity-theoretic question is how much we can make $c$ 
small (ideally $c =0$, which matches the universal lower bound). Unfortunately, 
achieving such an universal bound is an impossible goal for many problems, e.g., minimum spanning tree (MST), shortest paths, 
minimum cut, and so on. They exhibit the lower bound of 
$\tilde{\Omega}(\sqrt{n} + D)$ rounds for general graphs. 

Most of $\tilde{\Omega}(n^c + D)$-round lower bounds for some $c > 0$ are 
derived from special ``hard-core'' instances, and does not necessarily apply to 
popular graph classes such as planar graphs, which evokes the interest of  
developing efficient distributed graph algorithms for specific graph classes.
In the last few years, the study along this line rapidly made progress, where
the concepts of \emph{partwise aggregation} and \emph{low-congestion shortcuts} play an important role. In the partwise aggregation problem, all the nodes 
in the network is initially partitioned into a number of disjoint connected subgraphs, which we call a \emph{part}. The goal of this problem is to perform a certain kind of distributed tasks independently within all the parts in parallel. The executable tasks cover several standard operations such as 
broadcast, convergecast, leader election, finding minimum, and so on. 
The \emph{low-congestion shortcut} is a framework of solving the partwise aggregation problem, which is initiated by Ghaffari and Haeupler\cite{GH16}.
The key difficulty of the partwise aggregation problem appears when 
the diameter of a part is much larger than the diameter $D$ of the original 
graph. Since the diameter can become $\Omega(n)$ in the worst case, the 
naive solution which performs the aggregation task only by in-part 
communication can cause the expensive $\Omega(n)$-round running time. A low-congestion shortcut is 
defined as the sets of links augmented to each part for accelerating the aggregation task there. Its efficiency is characterized by two quality 
parameters: The \emph{dilation} is the maximum 
diameter of all the parts after the augmentation, and the 
\emph{congestion} is the maximum edge congestion of all edges $e$, where
the edge congestion of $e$ is defined as the number of the parts augmenting $e$. 
In the application of low-congestion 
shortcuts, the performance of an algorithm typically relies on the sum of the dilation and congestion. Hence we simply call the value of 
dilation plus congestion the \emph{quality} of the shortcut. It is known that 
any low-congestion shortcut with quality $q$ and $O(f)$-round construction time
yields an $\tilde{O}(f + q)$-round solution for the partwise aggregation problem,
and $\tilde{O}(f + q)$-round partwise aggregation yields the efficient solutions
for several fundamental graph problems. Precisely, the following meta-theorem holds.
\begin{theorem}[Ghaffari and Haeupler\cite{GH16}, Haeupler and Li\cite{H18}]
\label{thm:shortcut-application}
Let $\mathcal{G}$ be a graph class allowing the low-congestion shortcut with quality $O(q)$ that can be constructed in $O(f)$ rounds in the CONGEST model.
Then there exist three algorithms solving (1) the MST problem in 
$\tilde{O}(f + q)$ rounds, (2) the $(1 + \epsilon)$-approximate minimum cut problem
in $\tilde{O}(f + q)$ rounds for any $\epsilon = \Omega(1)$, and (3) 
$O(n^{O(\log\log n) / \log \beta})$-approximate weighted single-source shortest path problem in $\tilde{\Omega}((f + q)\beta)$ rounds for any $\beta = 
\Omega(\mathrm{polylog}(n))$\footnote{The statement of the weighted single-source shortest path problem is slightly simplified. See \cite{H18} for the details.}.
\end{theorem}
Conversely, if we get a time-complexity lower bound for any problem stated above, then it also applies to the partwise aggregation and low-congestion shortcuts (with 
respect to quality plus construction time). In fact, the 
$\tilde{O}(\sqrt{n} + D)$-round lower bound  of 
shortcuts for general graphs is deduced from  the lower bound of MST. 
On the other hand, the existence of efficient (in the sense of breaking the general lower bound) low-congestion shortcuts is known for several major graph classes, as well as its construction algorithms~\cite{haeupler2018minor,GL18,HIZ16,HIZ16-2,GH16,ghaffari2017distributed}.

\subsection{Our Result}
In this paper, we study the relationship between several 
major graph parameters and the quality of low-congestion shortcuts.
Specifically, we focus on three parameters, that is, (1) chordality, 
(2) diameter, and (3) clique-width. The precise statement of our result
is as follows:
\begin{itemize}
\item  There is an $O(1)$-round algorithm which constructs a low-congestion shortcut with quality $O(kD)$ for any $k$-chordal graph. When $k=O(1)$, its quality matches the $\Omega(D)$-universal lower bound.
\item For $k\leq D$ and $kD\leq \sqrt{n}$, there exists a $k$-chordal graph 
where the construction of MST requires $\tilde{\Omega}(kD)$ rounds.
It implies that the quality plus construction time of our algorithm is 
nearly optimal up to polylogarithmic factors.
\item There exists an algorithm
of constructing a low-congestion shortcut with quality $\tilde{O}(n^{1/4})$
in $\tilde{O}(n^{1/4})$ rounds for any graph of diameter three. 
In addition, there exists an algorithm of constructing a low-congestion shortcut with quality $\tilde{O}(n^{1/3})$ in $\tilde{O}(n^{1/3})$ rounds for any graph of diameter four. These results almost close the long-standing complexity gap of the MST construction in graphs with small diameters, which is originally posed by Lotker et~al.~\cite{LPP06}. 
\item We present a negative instance certifying that bounded clique-width does not help the construction of good-quality shortcuts.
Precisely, we give an instance of clique-width six where the construction of MST 
is as expensive as the general case, i.e., $\tilde{\Omega}(\sqrt{n}+D)$ rounds.
\end{itemize}
Table~\ref{fig:result} summarizes the state-of-the-art upper and lower 
bounds for low-congestion shortcuts. It should be noted
that all the parameters considered in this paper is independent of the other
parameters such that bounding it admits good shortcuts (e.g., 
treewidth and genus), and thus any result above is not a corollary of
the past results.

For proving our upper bounds, we propose a new scheme of 
shortcut construction, called \emph{1-hop extension}, where
each node in a part only takes all the incident edges as the shortcut edges 
of its own part. Surprisingly, this very simple construction admits an 
optimal shortcut for any $k$-chordal graph. For graphs of diameter three or 
four, our algorithm is obtained by combining
the 1-hop extension scheme with yet another algorithm of finding short 
low-congestion paths (i.e., paths of length one or two) connecting two
moderately-large subgraphs. These algorithms are still simple but it is 
far from triviality to bound the quality of constructed shortcuts. The 
analytic part includes several (seemingly) new ideas and may be of 
independent interest. 
\begin{table}[htbp]
\centering
\caption{The quality bounds of Low-Congestion Shortcuts for Specific Graph Classes}
\label{fig:result}
\begin{threeparttable}
\begin{tabular}{|l c c c|}
\hline
Graph Family & Quality & Construction & Lower bound\\
\hline
General & $\tilde{O}(\sqrt{n}+D)$ \cite{SD98} & $\tilde{O}(\sqrt{n}+D)$ \cite{SD98}& $\Omega(\sqrt{n}+D)$ \cite{PR00}\\
Planar &  $\tilde{O}(D)$ \cite{GH16}& $\tilde{O}(D)$ \cite{GH16}& $\tilde{\Omega}(D)$ \cite{GH16}\\
Genus-$g$ & $\tilde{O}(\sqrt{g}D)$ \cite{HIZ16-2}& $\tilde{O}(\sqrt{g}D)$ \cite{HIZ16-2}&$\tilde{\Omega}(\sqrt{g}D)$ \cite{HIZ16-2}\\
Treewidth-$k$ & $\tilde{O}(kD)$ \cite{HIZ16-2}& $\tilde{O}(kD)$ \cite{HIZ16-2}&$\Omega(kD)$ \cite{HIZ16-2}\\
Clique-width-6& -- & -- &$\tilde{\Omega}(\sqrt{n}+D)$ (\textbf{this paper})\\
Expander & $\tilde{O}\left(\tau2^{O\left(\sqrt{\log n}\right)}\right)$ \cite{GL18}\tnote{*}& $\tilde{O}\left(\tau2^{O\left(\sqrt{\log n}\right)}\right)$ \cite{GL18}& -- \\
$k$-Chordal & $O(kD)$ (\textbf{this paper}) & $O(1)$ (\textbf{this paper})& $\tilde{\Omega}(kD)$ (\textbf{this paper})\\
Excluded Minor & $\tilde{O}(D^{2})$ \cite{haeupler2018minor} & $\tilde{O}(D^2)$ \cite{haeupler2018minor}&  -- \\
$D=3$ & $\tilde{O}(n^{1/4})$ (\textbf{this paper}) & $\tilde{O}(n^{1/4})$ (\textbf{this paper})& $\Omega(n^{1/4})$ \cite{lowerbound,LPP06}\\
$D=4$ & $\tilde{O}(n^{1/3})$ (\textbf{this paper}) & $\tilde{O}(n^{1/3})$ (\textbf{this paper})& $\Omega(n^{1/3})$ \cite{lowerbound,LPP06}\\
$5 \leq D\leq \log n$ & -- & -- & $\tilde{\Omega}\left(n^{(D-2)/(2D-2)}\right)$ \cite{lowerbound}\\
\hline
\end{tabular}
\begin{tablenotes}
\item[*]$\tau$ is the mixing time of the network graph $G$.
\end{tablenotes}
\end{threeparttable}
\end{table}
\subsection{Related Work}
The MST problem is one of the most fundamental problems in distributed graph algorithms. It is not only important by itself, but also has many applications
for solving other distributed tasks (e.g., detecting connected components, 
minimum cut, and so on). Hence many researches have tackled the design of efficient MST algorithms in the CONGEST model so far~\cite{GHS83,SD98,garay1998sublinear,pandurangan2017time,pandurangan2018distributed,GP2018,ghaffari2018distributed,haeupler2018round,jurdzinski2018mst}.
The round-complexity lower bound of MST construction is also a central topic in distributed complexity 
theory~\cite{PR00,lowerbound,LPP06,ookawa2015filling,Elkin04,elkin2006unconditional}.
The inherent difficulty of MST construction is 
of solving the partwise aggregation (minimum) problem efficiently. This viewpoint
is first identified by Ghaffari and Haeupler~\cite{GH16} explicitly, as well as an efficient algorithm for solving it in planar graphs. The concept of 
low-congestion shortcuts is newly invented there for encapsulating the 
difficulty of partwise aggregation. Recently, several follow-up papers are
published to extend the applicability of low-congestion shortcuts, which 
break the known general lower bounds of several fundamental graph problems 
in several specific graph classes: This line includes bounded-genus 
graphs~\cite{GH16,HIZ16}, bounded-treewidth graphs\cite{HIZ16}, graphs with 
excluded minors~\cite{haeupler2018minor}, expander graphs~\cite{ghaffari2017distributed,GL18}, and so on (See Table~\ref{fig:result}).

The application of low-congestion shortcuts is not limited only to MST. As stated in Theorem~\ref{thm:shortcut-application}, it also admits efficient solutions for approximate minimum cut and single-source shortest path. A few algorithms recently proposed utilize low-congestion shortcuts as an important building block,
e.g., the depth first search in planar graphs~\cite{H18} and approximate treewidth
(with decomposition)~\cite{li2018distributed}. 
Haeupler et al.~\cite{haeupler2018round} shows a message-reduction scheme of
shortcut-based algorithms, which drop the total number of messages exchanged
by the algorithm into $\tilde{O}(m)$, where $m$ is the number of links.
On the negative side, it is known that the hardness of (approximate) diameter
cannot be encapsulated by low-congestion shortcuts. Abboud et al.~\cite{abboud2016near} shows a hard-core family of unweighted graphs with $O(\log n)$ treewidth where any diameter computation in the CONGEST model requires $\tilde{\Omega}(n)$ rounds.
Since any graph with $O(\log n)$ treewidth admits a low-congestion shortcut of
quality $\tilde{O}(D)$, this result implies that it is not possible to compute 
the diameter of graphs efficiently by using only the property of low-congestion
shortcuts.
 
While our results exhibit a tight upper bound for graphs of diameter three or
four, a more generalized lower bound is known for small-diameter 
graphs.~\cite{lowerbound}. For any $\log n \geq D \geq 3$, 
it is proved that there exists a network topology which incurs the 
$\tilde{\Omega}\left(n^{(D-2)/(2D-2)}\right)$-round time complexity for
any MST algorithm. In more restricted cases of $D = 1$ and $D= 2$, 
Jurdzinski et al.~\cite{jurdzinski2018mst} and Lotker et al.~\cite{LPP06} respectively show $O(1)$-round and $O(\log n)$-round MST algorithms.

\subsection{Outline of the Paper}
The paper is organized as follows: In Section \ref{sec:preliminaries}, we introduce the formal definitions of the CONGEST model,
 partwise aggregation, and low-congestion shortcuts, and other miscellaneous terminologies and notations. In Section \ref{sec:k-chordal}, we show the upper and lower bounds for shortcuts and MST in $k$-chordal graphs. In Section \ref{sec:small_diameter}, we present our shortcut algorithms for graphs of diameter three or 
four. In Section \ref{sec:clique-width}, we prove the hardness result for bounded clique-width graphs. The paper is concluded in Section \ref{sec:conclusion}.

\section{Preliminaries}
\label{sec:preliminaries}
\subsection{CONGEST model}
Throughout this paper, we denote by $[a,b]$ the set of integers at least $a$ and at most $b$. A distributed system is represented by a simple undirected connected graph $G=(V,E)$, where $V$ is the set of nodes and $E$ is the set of edges.
Let $n$ and $m$ be the numbers of nodes and edges respectively, and $D$
be the diameter of $G$. 
Each node has an \emph{ID} from $\mathbb{N}$ (which is represented with 
$O(\log n)$ bits). In the CONGEST model, the computation follows the round-based synchrony. In one round, each node sends messages to its neighbors, receives messages from its neighbors, and executes local computation. It is guaranteed that every message sent at a round is delivered to the destination within the same round. Each link can transfer $O(\log n)$-information (bidirectionally) per one round, and each node can inject different messages to its incident links. Each node has no prior knowledge on the network topology
except for its neighbor's IDs. Given a graph $H$ for which the node and link sets are not explicitly specified, we denote them by $V_H$ and $E_H$ respectively. Let $N(v)$ be the set of nodes that are adjacent to $v$, and 
$N^{+}(v) = N(v)\cup \{v\}$. We define 
$N(S) = \cup_{s \in S} N(s)$ and $N^{+}(S) = \cup_{s \in S} N^+(s)$
for any $S \subseteq V$. For two node subsets $X, Y \subseteq V$, we also define 
$E(X, Y) = \{(u, v) \in E \mid u \in X, v \in Y\}$. If $X$ (resp. $Y$) is a 
singleton $X = \{w\}$, (resp. $Y=\{w\}$), we describe $E(X, Y)$ as $E(w, Y)$ (resp. $E(X, w)$). The \emph{distance} (i.e., the number of edges in 
the shortest path) between two nodes $u$ and $v$ in $G$ is denoted by
$\Dist_{G}(u,v)$. Let $S$ be a path in $G$. With a small abuse of notations, we often treat $S$ as the sequence of nodes or edges representing the path, as 
the set of nodes or edges in the path, or the subgraph of $G$ forming the path.  

\subsection{Partwise Aggregation}
The \emph{partwise aggregation} is a communication abstraction defined 
over a set $\mathcal{P} = \{P_1, P_2, \dots, P_N\}$ of mutually-disjoint 
and connected subgraphs called \emph{parts}, and provides simultaneous 
fast communication among the nodes in each $P_i$. It is formally 
defined as follows:
\begin{definition}[Partwise Aggregation (PA)] 
  Let $\mathcal{P} = \{P_1, P_2, \dots, P_N\}$ be the set of 
  connected mutually-disjoint subgraphs of $G$, and 
  each node $v \in V_{P_i}$ maintains variable $b^i_v$ storing an
  input value $x^i_v \in X$. The output of the partwise aggregation problem 
  is to assign $\oplus_{w \in P_i} x^i_w$ with $b^i_v$ 
  for any $v \in V_{P_i}$, where $\oplus$ is an arbitrary associative and commutative binary operation over $X$.
\end{definition}
The straightforward solution of the partwise aggregation problem is to 
perform the convergecast and broadcast in each part $P_{i}$ independently.
Specifically, we construct a BFS tree for each part $P_i$ (after the selection of 
the root by any leader election algorithm). The time complexity is proportional to the diameter of each part $P_i$, which can be large ($\Omega(n)$ in the worst 
case) independently of the diameter of $G$. 
%
\subsection{$(d,c)$-Shortcut}
As we stated in the introduction, the notion of low-congestion shortcuts is
introduced for quickly solving the partwise aggregation problem 
(for some specific graph classes). The formal definition of $(d, c)$-shortcuts 
is given as follows.
\begin{definition}\label{def:shortcut}[Ghaffari and Haeupler\cite{GH16}]
Given a graph $G=(V,E)$ and a partition $\mathcal{P} = \{P_1, P_2, \dots, P_N\}$ 
of $G$ into node-disjoint and connected subgraphs, we define a $(d,c)$-shortcut
of $G$ and $\mathcal{P}$ as a set of subgraphs $\mathcal{H} = 
\{H_1, H_2, \dots, H_N\}$ of $G$ such that:
\begin{enumerate}
\item For each $i$, the diameter of $P_{i}+H_i$ is at most $d$ ($d$-dilation).
\item For each edge $e\in E$, the number of subgraphs $P_{i}+H_{i}$ containing $e$ is at most $c$ ($c$-congestion).
\end{enumerate}
\end{definition}
The values of $d$ and $c$ for a $(d,c)$-shortcut $\mathcal{H}$ is 
called the \emph{dilation} and \emph{congestion} of $\mathcal{H}$. As a general
statement, a $(d, c)$-shortcut which is constructed in $f$ rounds admits
the solution of the partwise aggregation problem in $\tilde{O}(d+c+f)$ 
rounds~\cite{GH16,Ghaffari15}. Since the parameter $d + c$ asymptotically 
affects the performance of the application, we call the value of $d + c$
the \emph{quality} of $(d, c)$-shortcuts. A low-congestion shortcut 
with quality $q$ is simply called a \emph{$q$-shortcut}.

\subsection{The framework of the Lower Bound}
\label{sec:mst_lowerbound}
\begin{figure}[t]
\begin{center}
 \includegraphics[width=100mm]{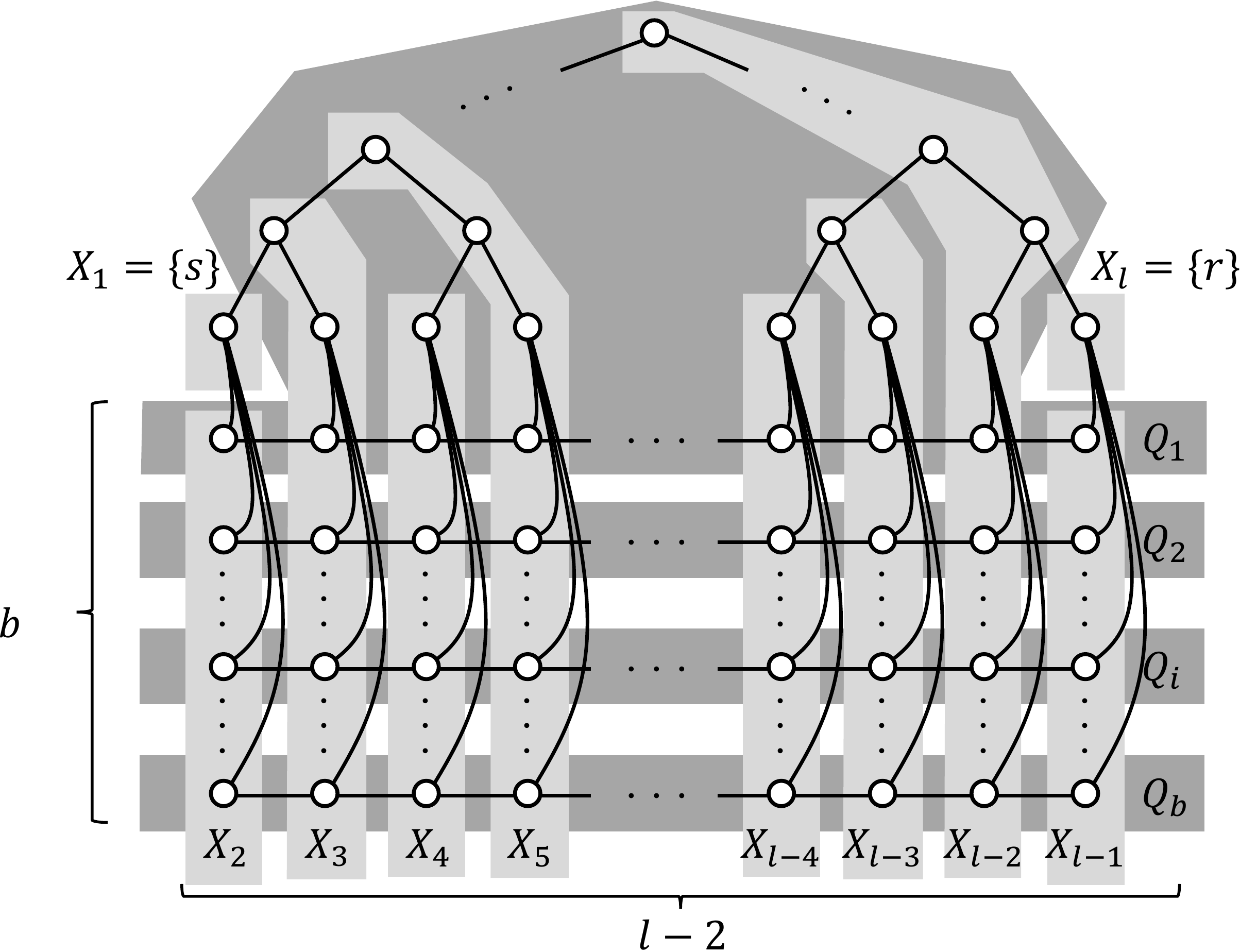}
 \end{center}
 \caption{Example of $\mathcal{G}(O(lb),b,l,O(\log n))$}
 \label{fig:lowerMST}
\end{figure}
To prove the lower bound of MST, we introduce a simplified version of the 
framework by Das Sarma et al.~\cite{lowerbound}. In this framework, 
we consider the graph class $\mathcal{G}(n,b,l,c)$ that is defined below.
A vertex set $X\subseteq V$ is called \emph{connected} if the subgraph induced
by $X$ is connected.
\begin{definition}
For $n,b,c\geq 0$ and $l\geq 3$, the graph class $\mathcal{G}(n,b,l,c)$ is defined as the set of $n$-vertex graph $G = (V, E)$ satisfying the following conditions:
\begin{itemize}
\item (\textbf{C1}) The vertex set $V$ is partitioned 
into $\ell$ disjoint vertex sets 
$\mathcal{X} = \{X_1, X_2, \dots, X_{\ell}\}$ such that $X_1$ and $X_{\ell}$
are singletons (let $X_1 = \{s\}$ and $X_{\ell} = \{r\}$).
\item (\textbf{C2}) The vertex set $V \backslash \{s, r\}$ is partitioned into $b$ disjoint connected sets $\mathcal{Q} = \{Q_1,\dots,Q_{b}\}$ such that 
$|E(X_1, Q_i)| \geq 1$ and $|E(X_{l}, Q_i)| \geq 1$ hold for any 
$1\leq i \leq b$.
\item (\textbf{C3}) Let $R_{i}=\bigcup_{i+1 \leq j \leq l}X_{j}$ and $L_{i}=\bigcup_{0 \leq j \leq l-1-i}X_{j}$.
For $2 \leq i \leq l/2-1$, $|E(R_{i}, N(R_{i})\setminus R_{i-1})| \leq c$ and $|E(L_{i}, N(L_{i})\setminus L_{i-1})| \leq c$.
\end{itemize}
\end{definition}
Figure~\ref{fig:lowerMST} shows the graph that is defined vertex partition $\mathcal{X}$ and $\mathcal{Q}$ for the hard-core instances presented 
in the original proof by Das Sarma et al.~\cite{lowerbound}.
This graph belongs to $\mathcal{G}(O(lb),b,l,O(\log n))$.
For class $\mathcal{G}(n,b,l,c)$, the following theorem holds, which is just
a corollary of the result by Das Sarma et al.~\cite{lowerbound}.
\begin{theorem}[Das Sarma et al.\cite{lowerbound}]
\label{thm:communication_complexity}
For any graph $G \in \mathcal{G}(n,b,l,c)$ and any MST algorithm $A$, 
there exists an edge-weight function $w_{A,G}:E \to \mathbb{N}$ such that 
the execution of $A$ in $G$ requires $\tilde{\Omega}(\min\{b/c,l/2-1\})$ rounds.
This bound holds with high probability even if $A$ is a randomized algorithm.
\end{theorem}
%
\section{Low-Congestion Shortcut for $k$-Chordal Graphs}
\label{sec:k-chordal}
\subsection{$k$-Chordal Graph}
A graph $G$ is $k$-\emph{chordal} if and only if every cycle of length larger
than $k$ has a chord (equivalently, $G$ contains no induced cycle of length 
larger than $k$). In particular, $3$-chordal graphs are simply called \emph{chordal} graphs, which is known to be much related to various intersection graph 
families such as interval graphs\cite{gavril1974intersection,pal2014intersection}. Since $k$-chordal graphs can 
contain the clique of an arbitrary size for any $k \geq 3$, it is never
a subclass of any minor-excluded graphs. Thus no known shortcut algorithm 
works correctly for $k$-chordal graphs. The main results of 
this section are the following two theorems:
\begin{theorem}\label{chordal}
There is an $O(1)$-round algorithm which constructs a $O(kD)$-shortcut for any $k$-chordal graph. 
\end{theorem}
\begin{theorem}\label{kchordalkakai}
For $k\leq D$ and $kD\leq \sqrt{n}$, there exists an unweighted $k$-chordal 
graph $G = (V, E)$ where for any MST algorithm $A$, there exists 
an edge-weight function $w_A : E \to \mathbb{N}$ such that the running
time of $A$ becomes $\tilde{\Omega}(kD)$ rounds.
\end{theorem}

\subsection{Proof of Theorem~\ref{chordal}}
We provide the proof of Theorem~\ref{chordal}.
The construction algorithm is very simple. It follows the \emph{1-hop extension} scheme stated below:
\begin{quote}
For any $V_{P_i} \subseteq V$, node $v \in V_{P_{i}}$ adds each incident edge 
$(v, u)$ 
to $H_i$, and informs $u$ of the fact of $(v, u) \in H_i$.
\end{quote}
Obviously, this algorithm terminates in one round.
Since each node belongs to one part, the congestion of each edge is at most two.
Therefore, the technical challenge in proving Theorem~\ref{chordal} is to show that the diameter of $P_i+H_i$ is $O(kD)$ for any $i \in [1, N]$. In other words,
the following lemma trivially deduces Theorem~\ref{chordal}. 
\begin{lemma}\label{kchordald}
Letting $G_i=P_i + H_{i}$, $\Dist_{G_i}(u, v) \leq kD + 2$ holds for any 
$u, v \in V_{G_i}$.
\end{lemma}
%
%

\begin{proof}
We show that $\Dist_{G_i}(u, v) \leq kD$ holds for any $u, v \in V_{P_i}$.
Since any node in $v \in V_{G_i} \setminus V_{P_i}$ is a neighbor of a node in $V_{P_i}$, it obviously follows the lemma. 

Let $A$ be the shortest path from $u$ to $v$ in $G$, and $B$ be that in $P_i$.
We define $T=(t_0,t_1,\dots,t_{z-1})$ as the sequence of nodes in $A \cap B$ 
which are sorted in the order of $A$. By definition, $u = t_0$ and 
$v = t_{z-1}$ holds. The core of the proof is to show 
that $\Dist_{G_i}(t_x,t_{x+1}) \leq k\cdot\Dist_{G}(t_x, t_{x+1})$ for 
$0 \leq x \leq z-1$. Summing up this inequality for all $x$,
we obtain $\Dist_{G_i}(t_0,t_{z-1}) \leq \sum_{1 \leq j \leq z} 
k\Dist_{G}(t_{j-1},t_{j}) = kD$. By symmetry, we only consider the case
of $x = 0$. The case of $x > 0$ is proved similarly.
Let $S = (t_0 = s_0, s_1, \dots, s_{\ell} = t_1)$ be the sub-path of 
$A$, and $S' = (t_0 = s'_0, s'_1, \dots, s'_{\ell'} = t_1)$ be the sub-path of 
$B$. Given a sequence $X$, we denote by $X[i, j]$ its consecutive subsequence 
from the $i$-th element to the $j$-th one in $X$.

We prove 
that for any $0 \leq j \leq \ell$, there exists a node $s_{c(j)} \in S$
such that $c(j) \geq j$, $\Dist_{G_i}(t_0 ,s_{c(j)}) \leq kj$ and $N^{+}(s_{c(j)})\cup S'\neq \emptyset$ hold.
The lemma is obtained by setting $j= \ell$ because then
$s_{c(j)} = s_{\ell} = t_1$ holds. The proof follows the induction on $j$.  
(Basis) If $j=0$, then it holds for $s_{c(j)} = s_0$. (Inductive step) 
Suppose as the induction hypothesis that there exists a node 
$s_{c(j)}$ satisfying $c(j) \geq j$ and $\Dist_{G_i}(t_0, s_{c(j)}) \leq kj$.
If $c(j) > j$, obviously $s_{c(j+1)} = s_{c(j)}$ satisfies the case of $j+1$. 
Thus, it suffices to consider the case of $c(j) = j$. Let $s'_h$ be the
neighbor of $s_{c(j)}$ in $S'$ maximizing $h$, and 
$e = (s_{c(j)}, s'_h)$. We consider the cycle $C$ consisting of 
$S[c(j), \ell]$, $S'[h, \ell']$, and $e$. If the length of $C$ is at most $k$,
obviously we have $\ell' - h \leq k - 1$. Since 
$\Dist_{G_i}(t_0, s_{c(j)}) \leq kj$ holds by the induction hypothesis, 
$s_{c(j+1)} = s_\ell$ satisfies the condition. If the length of $C$ is 
larger than $k$, $C$ has a chord, which connects two nodes respectively 
in $S$ and $S'$ because both $S$ and $S'$ are shortest paths. Let 
$e'=(s_{y},s'_{y'})$ be such a chord making the cycle $C'$ consisting of 
$e$, $e'$, $S[s_{c(j)}, s_y]$, and $S'[s_h, s_y']$ chordless (see 
Figure~\ref{fkchordalpart}). Since $h$ is the maximum, we have $y > c(j)$
because if $y = c(j)$ the edge $e' (\neq e)$ is taken as $e$. 
Due to the property of $k$-chordality, the length of $C'$ is at most $k$, 
and thus the length of path $S'[h, y'] + \{e, e'\}$ from $s_{c(j)}$ to 
$s_{c(x)+y}$ is at most $k-1$, that is, $\Dist_{G_i}(s_{c(j)},s_{y}) 
\leq k$. By the induction hypothesis, we obtain $\Dist_{G_i}(t_0, s_{y})\leq\Dist_{G_i}(t_{0},s_{c(j)})+\Dist_{G_{i}}(s_{c(j)},s_y) \leq k(j+1)$. Since $s'_{y'}$ is the neighbor of $s_{y}$, we have $N^{+}(s_{y})\cup S' \neq \emptyset$. 
Letting $c(j+1) = y$, we obtain the proof for $j + 1$. The lemma holds.
\end{proof}

\begin{figure*}[t]
\begin{center}
 \includegraphics[width=100mm]{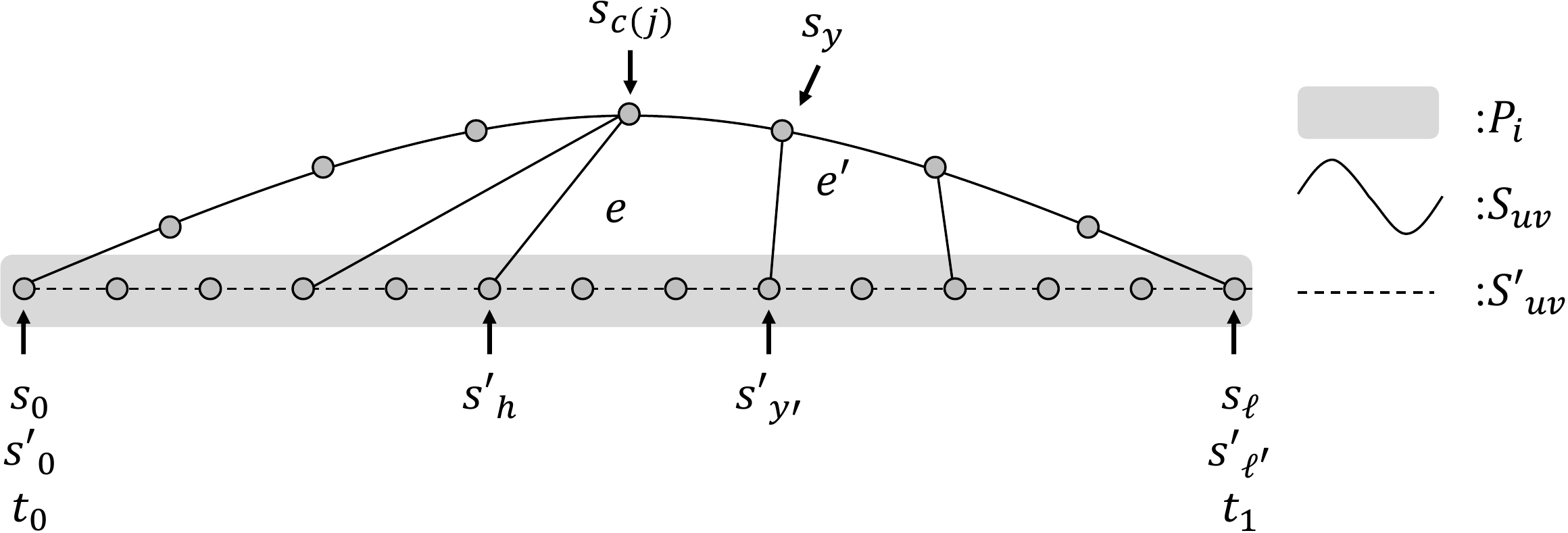}
 \end{center}
 \caption{Proof of Lemma~\ref{kchordald}.}
 \label{fkchordalpart}
\end{figure*}
\subsection{Proof of Theorem~\ref{kchordalkakai}}
We first introduce the instance mentioned in Theorem~\ref{kchordalkakai}. Since it has two additional parameters $x\geq 0$ and $N\geq 2$ as well as $k$, we refer to that instance as $G(k,x,N) = (V(k, x, N), E(k,x,N))$
in the following argument. The parameters $ x $ and $ N $ are adjusted later for obtaining the claimed lower bound. Let $K = k/2 - 1$ for short. The vertex set
and edge set of $G(k, x, N)$ is defined as follows:
\begin{itemize}
    \item $V(k, x, N) = \{v_{1,j} \mid 0 \leq j \leq x\} \cup \{v_{i,j} | 2 \leq i \leq N , 0 \leq j \leq xK \}$.
    \item $E(k, x, N) = E_1 \cup E_2 \cup E_3 \cup E_4$ such that  
    $E_1 = \{\{v_{1,j},v_{1,j+1}\} \mid 0 \leq j \leq x-1\}$, 
    $E_2 = \{\{v_{i,j},v_{i,j+1}\} \mid 2 \leq i \leq N , 0 \leq j \leq xK - 1\}$, $E_3 = \{\{v_{1,j},v_{i,h}\} \mid 2 \leq i \leq N , 0 \leq j \leq x, h = jK\}$, and $E_4 = \{\{v_{i,h},v_{j,h}\} \mid 2 \leq i, j \leq N, i \neq j, h\bmod{K} = 0\}$.  
\end{itemize}
Figure~\ref{kchordalkk} illustrates the graph $G(k, x, N)$.
\begin{figure*}[t]
\begin{center}
 \includegraphics[width=100mm]{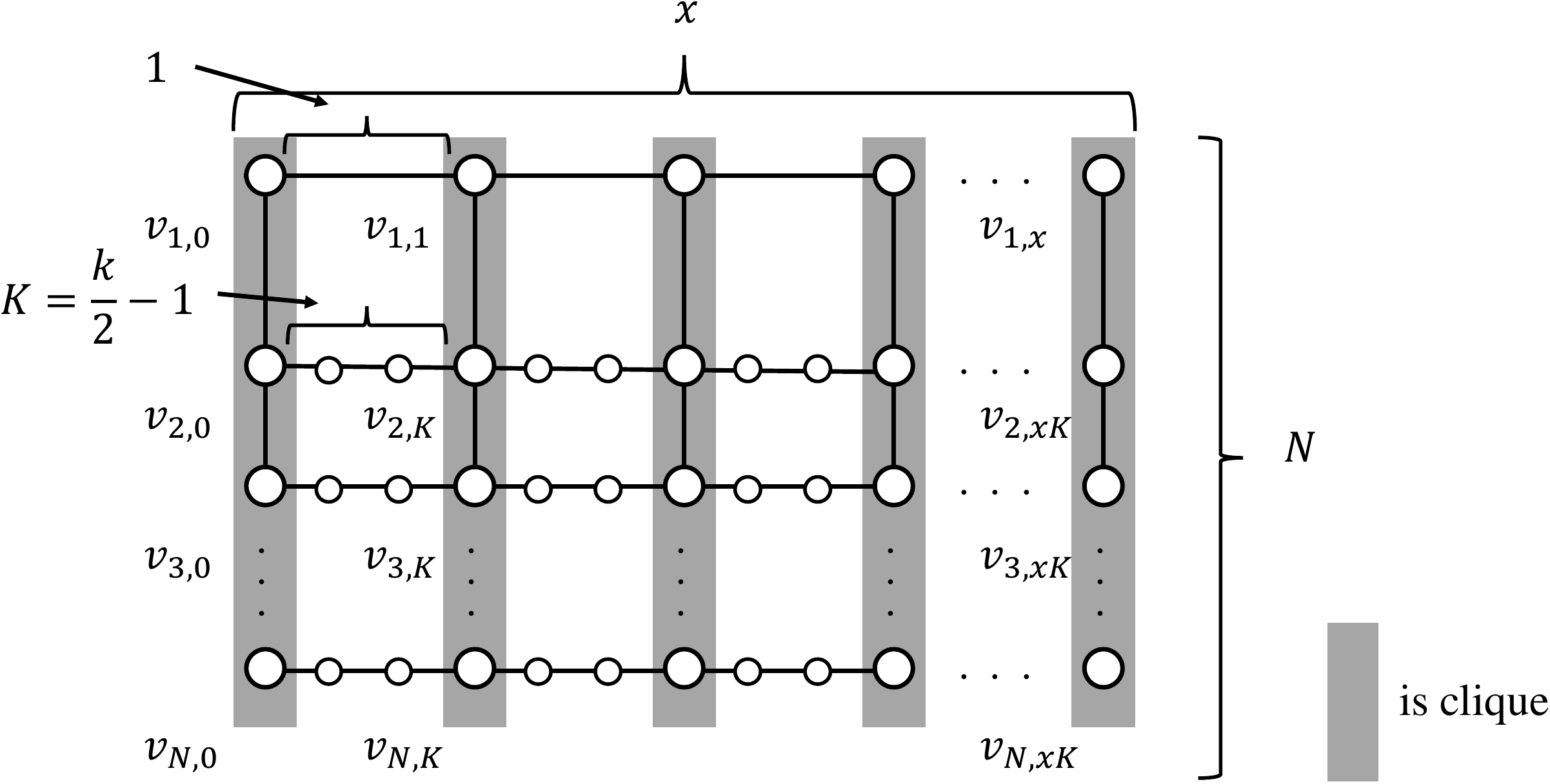}
 \end{center}
 \caption{Example of $k$-chordal graph $G(k,x,N)$.}
 \label{kchordalkk}
\end{figure*}
It is cumbersome to check this graph is $k$-chordal, but straightforward.
One can show the following lemma.
\begin{lemma}\label{circlehodai}
    For $x\geq0$ and $N\geq 2$, $G(k,x,N)$ is $k$-chordal.
\end{lemma}
\begin{proof}
For simplicity, we give some of the vertices a name $v'_{xy}$ as follows;
\begin{itemize}
\item $v'_{1,j} = v_{1,j} (0 \leq j \leq x)$
\item $v'_{i,j} = v_{i,h} (2 \leq i \leq N , 0 \leq j \leq x , h = jK)$.
\end{itemize}
We define a subset of vertices called \emph{row} and \emph{column}.
The $i$-th row $R_i$ is defined as $R_i=\{v'_{i,j} | 0 \leq j \leq x\}$, and 
the $i$-th column $C_i$ is defined as $C_i=\{v'_{j,i} | 1 \leq j \leq N\}$.

First, we consider the diameter of $G(k,x,N)$.
For $2\leq i \leq N$ and $0\leq j \leq xK$, we have
$\min_{0\leq k \leq x}\mathit{dist}(v'_{1,k},v_{i,j}) =\min_{0\leq k \leq x}\mathit{dist}(v'_{i,k},v_{i,j})+1 \leq K/2 + 1$.
For $0\leq i\leq x$ and $0\leq j\leq x$, $\mathit{dist}(v'_{1,i},v_{1,j})\leq x-1$,
holds and thus the diameter of $G(k,x,N)$ is at most $K+1+x$.

We consider a cycle $X$ in G$(k, x, N)$.
Let $l$ and $r$ be the minimum/maximum indices of the rows $X$ intersects, 
Similarly, let $t$ and $b$ be the minimum/maximum indices of the columns $X$ intersects. Let $m$ be the index such that $|C_m \cap X|$ maximizes, and
let $a_m = |C_m \cap X|$ for short.
Any cycle $X$ applies to one of the following four cases.
\begin{enumerate}
\item $r-l \geq 2$ holds.
\item $a_m \geq 3$ and $r-l\neq 0$ hold.
\item $r-l = 0$ holds.
\item $r-l = 1$ and $a_m = 2$ hold.
\end{enumerate}
We show that Lemma~\ref{circlehodai} holds for all the cases 
(Figure~\ref{proofkc} almost states the proof).
\begin{figure*}[t]
\begin{center}
 \includegraphics[width=120mm]{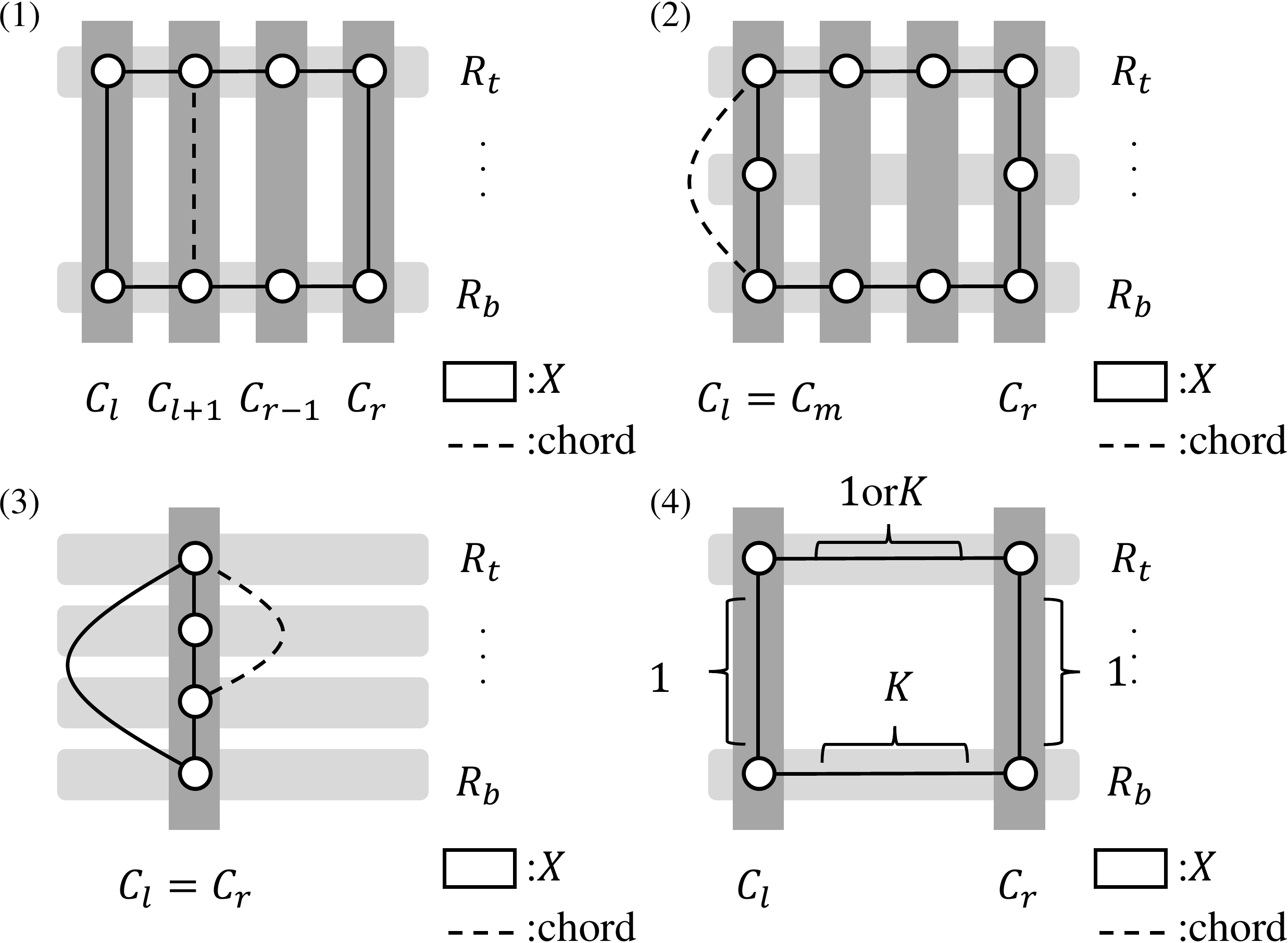}
 \end{center}
 \caption{Proof of Lemma~\ref{circlehodai}.}
 \label{proofkc}
\end{figure*}
\begin{enumerate}
\item The case of $r-l \geq 2$: By the construction of $G(k,x,N)$, $l$-$r$ path intersects $(l+1)$-column at least twice.
Let $u$ and $v$ be the intersection of $X$ and $(l+1)$-column.
Since $C_{l+1}$ is clique, $u$ and $v$ are adjacent.
Thus the edge $(u,v)$ is chord of $X$.
\item The case of $a_m \geq 3$ and $r-l\neq 0$: There exists two vertices in $C_{m}$, which are not adjacent in $X$.
Since $C_m$ is clique, there exists an edges between them, and this edge is a 
chord of $X$.
\item The case of $r-l = 0$: The cycle $X$ is a clique in graph $G$ and the lemma holds obviously.
\item The case of $r-l = 1$ and $a_m=2$: The cycle consists of four vertices $v'_{t, l}$,$v'_{t, r}$,$v'_{b, l}$, $v'_{b, r}$ and two paths, that is, the paths connecting $v'_{t, l}$ with $v'_{t, r}$, and $v'_{b, l}$ with $v'_{b, r}$.
It follows $\mathit{dist}(v'_{t, l},v'_{t, r}) \leq K = k/2 - 1$, $\mathit{dist}(v'_{b ,l},v'_{b, r}) = K = k/2 - 1$, and $\mathit{dist}(v'_{t, l},v'_{b, l}) = \mathit{dist}(v'_{t, r},v'_{b, r}) = 1$. Thus the length of $X$ is at most $k$.
\end{enumerate}
The lemma is proved.
\end{proof} 
The proof of Theorem~\ref{kchordalkakai} follows the framework by Das Sarma et al.\cite{lowerbound}. It suffices to show that the following lemma. Theorem~\ref{kchordalkakai} is obtained by combining this lemma with 
Theorem~\ref{thm:communication_complexity}.
\begin{lemma}\label{k-chordal_lowerbound}
    For any $D>2K$ and $N\geq 2kD$, $G(k,D-K,N) \in \mathcal{G}(n,N,(D-K)K+3,1)$
    holds.
\end{lemma}

\begin{proof}
We define $\mathcal{X}$ and $\mathcal{Q}$ for $G(k,D-K,N)$ as follows:
\begin{align*}
\mathcal{X} &= \{X_1,X_2,\dots,X_{(D-K)K+3}\} \ \ \text{s.t.} \\ 
X_{i}&=
\begin{cases}
        \left\{v_{1,0}\right\}& \hspace*{12mm} \text{$\left(i=1\right)$}\\
        \left\{v_{j,0}\mid 2\leq j \leq N\right\}& \hspace*{12mm} \text{$\left(i=2\right)$}\\
        \left\{v_{j,i-2}\mid 2\leq j \leq N\right\}\cup \left\{v_{\frac{i-2}{K},1} \right\}& \hspace*{12mm} \text{$\left(3\leq i \leq (D-K)K, i \bmod{K} =2\right)$}\\
        \left\{v_{j,i-2}\mid 2\leq j \leq N\right\}& \hspace*{12mm} \text{$\left(3\leq i \leq (D-K)K, i \bmod{K} \neq2\right)$}\\
        \left\{v_{j,(D-K)K-1}\mid 2\leq j \leq N\right\}& \hspace*{12mm} \text{$\left(i=(D-K)K+2\right)$}\\
        \left\{v_{1,(D-K)}\right\}& \hspace*{12mm} \text{$\left(i=(D-K)K+3\right)$}.
\end{cases}\\
\end{align*}
\begin{align*}
\mathcal{Q} &= \{Q_1,Q_2,\dots,Q_{N}\} \ \ \text{s.t.} \\ 
Q_{i}&=
\begin{cases}
        \left\{v_{1,j}\mid 1\leq j \leq (D-K)-1 \right\}& \hspace*{12mm} \text{$\left(i=1\right)$}\\
        \left\{v_{i,j}\mid 0\leq j \leq (D-K)K\right\}& \hspace*{12mm} \text{$\left(2\leq i\leq N\right)$}.
\end{cases}\\
\end{align*}
It is easy to check (\textbf{C1}) and (\textbf{C2}) is satisfied.
Thus we only show that (\textbf{C3}) is satisfied.
We have $E(R_{i},N(R_{i})\backslash R_{i-1})$ and $E(L_{i},N(L_i)\backslash L_{i})$ as follows:
\begin{align*}
E(R_{i},N(R_{i})\backslash R_{i-1})&=&
\begin{cases}
        \left\{v_{1,0}\right\}& \hspace*{12mm} \text{$\left(i=2\right)$}\\
        \left\{v_{1,\left\lfloor\frac{i-1}{K}\right\rfloor}\right\}& \hspace*{12mm} \text{$\left(3\leq i\leq \frac{(D-K)K}{2},i\bmod{K}\neq 2\right)$}\\
        \emptyset& \hspace*{12mm} \text{$\left(3\leq i\leq \frac{(D-K)K}{2},i\bmod{K}= 2\right)$}.
\end{cases}\\
E(L_{i},N(L_{i})\backslash L_{i-1})&=&
\begin{cases}
        \left\{v_{1,D-K}\right\}& \hspace*{12mm} \text{$\left(i=2\right)$}\\
        \left\{v_{1, D-K-\left\lfloor\frac{i-2}{K}\right\rfloor}\right\}& \hspace*{12mm} \text{$\left(3\leq i\leq \frac{(D-K)K}{2},i\bmod{K}\neq 2\right)$}\\
        \emptyset& \hspace*{12mm} \text{$\left(3\leq i\leq \frac{(D-K)K}{2},i\bmod{K}= 2\right)$}.
\end{cases}
\end{align*}
Thus we have $|E(R_{i},N(R_{i})\backslash R_{i-1})| \leq 1$ and $|E(L_{i},N(L_{i})\backslash L_{i-1})|\leq 1$.
Therefore we can prove that the graph $G(k,D-K,N)$ is included in $\mathcal{G}(n,N,(D-K)K+3,1)$. 
\end{proof}
%
%
%
%
%
%
%
%
%
\section{Low-Congestion Shortcut for Small diameter Graphs}
\label{sec:small_diameter}
Let $\kappa_D = n^{(D - 2)/(2D - 2)}$ for short. Note that $\kappa_3 = n^{1/4}$
and $\kappa_4 = n^{1/3}$ hold.
The main result in this section is the theorem below.
\begin{theorem}
\label{theo:dia3exist}
For any graph of diameter $D \in \{3, 4\}$, there exists an algorithm of constructing low-congestion shortcuts with quality $\tilde{O}(\kappa_{D})$ in $\tilde{O}(\kappa_{D})$ rounds.
\end{theorem}
\subsection{Centralized Construction}
In the following argument, we use term ``whp. (with high probability)'' 
to mean that the event considered occurs with probability 
$1 - n^{-\omega(1)}$ (or equivalently $1 - e^{-\omega(\log n)}$). For 
simplicity of the proof, we treat any whp. event as if it necessarily occurs 
(i.e. with probability one). Since the analysis below handles only a polynomially-bounded number of whp. events, the standard union-bound 
argument guarantees that everything simultaneously occurs whp. That is, any consequence yielded by the analysis also occurs whp.
Since the proof is constructive, we first present the algorithms for $D = 3$ and $4$. They are described as a (unified) centralized algorithm, and the distributed 
implementation is explained later. Let $N'$ be the number of parts whose diameter
is more than $12\kappa_D\log^3 n$ (say \emph{large} part). Assume that 
$P_1, P_2, \dots, P_{N'}$ are large without loss of generality.
Since each part $P_i$ ($1 \leq i \leq N'$) contains at least 
$\kappa_D$ nodes, $N' \leq n / \kappa_D$ holds obviously.  
The proposed algorithm constructs the shortcut edges $H_i$ for each 
large part $P_i$ following the procedure below: 
\begin{enumerate}
    \item Each node $v \in V_{P_i}$ adds its incident edges to $H_i$ (i.e., compute
    the 1-hop extension).
    \item This step adopts two different strategies according to the value of $D$. ($D = 3$) Each node $u \in N^{+}(V_{P_i})$ adds each incident edge $(u, v)$ to $H_i$ with probability $1 / n^{1/2}$. ($D = 4$) Let 
    $\mathcal{Y} = [1, n^{1/3}/\log n]$. We first prepare an 
    $(n^{1/3}\log^3 n)$-wise independent hash function 
    $h: [0, N - 1] \times V \to \mathcal{Y}$\footnote{Let $X$  and $Y$ 
    be two finite sets. For any integer $k\ge 1$, a family of hash functions $\mathcal{H} = \{h_1, h_2, \dots, h_p\}$, where each $h_i$ is a function from $X$ to $Y$, is called \emph{$k$-wise independent} if for any distinct
    $x_1, x_2, \dots, x_k \in X$ and any $y_1, y_2, \dots y_k 
    \in Y$, a function $h$ sampled from
    $\mathcal{H}$ uniformly at random satisfies 
    $\Pr[\bigwedge_{1 \leq i \leq k} h(x_i) = y_i] = 1 / |Y|^k$. }. 
    Each node $u \in V$ adds each incident edge $(u, v)$ to $H_i$ 
    with probability $1 / h(u, i)$ if $v \in N^+(V_{P_i})$.
\end{enumerate}

We show that this algorithm provides a low-congestion shortcut of 
quality $\tilde{O}(\kappa_D)$. First, we look at the bound for congestion.
Let $H^1_i$ be the set of the edges added 
to $H_i$ in the first step, and $H^2_i$ be those in the second step. 
Since the congestion of 1-hop extension is negligibly small, 
it suffices to consider the congestion incurred by step 2. 
Intuitively, we can believe the congestion of $\tilde{O}(\kappa_D)$ 
from the fact that the expected congestion of each edge is  
$\tilde{O}(\kappa_D)$: Since the total number of 
large parts is at most $n / \kappa_D$, the expected congestion of 
each edge incurred in step 2 is $n / \kappa_D \cdot (1 / n^{1/2}) = 
O(n^{1/4})$ for $D = 3$, and $(n/\kappa_D)\sum_{y \in \mathcal{Y}} (1/y) \cdot 
(1/|\mathcal{Y}|) \leq (n/\kappa_D) \cdot (\log n / |\mathcal{Y}|) = \tilde{O}(n^{1/3})$ for $D = 4$.
\begin{lemma}
\label{lma:congestion_small_diameter}
The congestion of the constructed shortcut is $\tilde{O}(\kappa_D)$ whp.
\end{lemma}
\begin{proof}
It suffices to show that the congestion of any edge $e = (u, v)\in E$ is 
$\tilde{O}(\kappa_D)$ whp. For simplicity of the proof, we see an 
undirected edge $e = (u, v)$ as two (directed) edges $(u, v)$ and $(v, u)$,
and distinguish the events of adding $(u, v)$ to shortcuts by $u$ and 
that by $v$. That is, the former is recognized as 
adding $(u, v)$, and the latter as adding $(v, u)$. Obviously, 
the asymptotic bound holding for directed edge $(u, v)$ also holds for the corresponding undirected edge $(u, v)$ actually existing in $G$ (which is at 
most twice of the directed bound). Since the first step of the algorithm 
increases the congestion of each directed edge at most by one, it suffices 
to show that 
the congestion incurred by the second step is at most $\tilde{O}(\kappa_D)$.

Let $X_i$ be the indicator random variable for the event $(u, v) \in H^2_i$,
and $X = \sum_{i} X_i$. The goal of the proof is to show that $X = \tilde{O}
(\kappa_D)$ holds whp. The cases of $D = 3$ and $D = 4$ are proved separately. 
($D = 3$) Since at most $n/ \kappa_3$ large parts exist, 
we have $\mathbb{E}[X] \leq (n/\kappa_3) \cdot (1/n^{1/2}) = n^{1/4} = \kappa_3.$
The straightforward application of Chernoff bound to $X$ allows us to bound the congestion of $e$ by at most $2\kappa_3$ 
with probability $1 - e^{-\Omega(n^{1/4})}$.  
($D = 4$) Let $\mathcal{P}'$ be the subset of all large parts $P_j$ such that
$u \in N^+(P_j)$ holds. Consider an arbitrary partition of $\mathcal{P'}$ into
several groups with size at least $(n^{1/3}\log^3 n)/2$ and at most 
$n^{1/3}\log^3 n$. 
Let $q$ be the number of groups. Each group is identified by a number 
$\ell \in [1, q]$. We refer to the $\ell$-th group as $\mathcal{P}^{\ell}$. 
Fixing $\ell$, we bound the number of parts in $\mathcal{P}^{\ell}$ using 
$e = (u, v)$ as a shortcut edge. Let $Y_{i}$ be the value of $h(u, i)$. 
For $P_i \in \mathcal{P}^{\ell}$, the probability that $X_{i} = 1$ is $\Pr[X_{i} = 1] = \sum_{y \in \mathcal{Y}} \Pr[Y_{i} = y] 1/y = \mathit{Har}\left(|\mathcal{Y}|\right)/|\mathcal{Y}|,$
where $\mathit{Har}(x)$ is the harmonic number of $x$, i.e., 
$\sum_{1 \leq i \leq x} i^{-1}$.
Letting $X^{\ell} = \sum_{j \in P^{\ell}} X_{j}$, we have $\mathbb{E}[X^{\ell}] =
(|P^{\ell}|\mathit{Har}(|\mathcal{Y}|))/|\mathcal{Y}|$. Since 
$\mathit{Har}(x) \leq \log x$, we have $(|P^{\ell}|\log n)/|\mathcal{Y}| 
\geq \mathbb{E}[X^{\ell}] \geq |P^{\ell}|/|\mathcal{Y}| = 
(\log^4 n)/2$.
Since the hash function $h$ is $(n^{1/3}\log^3 n)$-wise independent, it 
is easy to check that 
$X_{ 1}, X_{ 2}, \dots, X_{ p^{\ell}}$ are independent.
We apply Chernoff bound to $X^{\ell}$, and obtain $\Pr[X^{\ell} \leq 2\mathbb{E}[X^{\ell}]] \geq 1 - e^{-\Omega(\mathbb{E}[X^{\ell}])} =  1 - e^{-\Omega(\log^4 n)}.$
It implies 
that for any $\ell$ at most $2\mathbb{E}[X^{\ell}]$ groups use $(u, v)$ as 
their shortcut edges. The total congestion of $(u, v)$ is obtained by 
summing up $2\mathbb{E}[X^{\ell}]$ for all $\ell \in [1, q]$, which results
in $\sum_{\ell} 2|P^{\ell}| \log n / |\mathcal{Y}| = 2|\mathcal{P}'|\log n / |\mathcal{Y}| = \tilde{O}(n^{1/3}).$
The lemma is 
proved.
\end{proof}
For bounding dilation, we first introduce several preliminary notions and terminologies. Given a graph $G = (V, E)$, a subset $S \subset V$ is called 
an \emph{$(\alpha, \beta)$-ruling set} if it satisfies that (1) for any $u, v \in S$, 
$\Dist_G(u, v) \geq \alpha$ holds, and (2) for any node $v \in V$, there exists
$u \in S$ such that $\Dist_G(v, u) \leq \beta$ holds. It is known that
there exists an $(\alpha, \alpha +1)$-ruling set for any graph $G$~\cite{AGLP89}. 
 Let $\hat{P}_i = P_i + H^1_i$ for short. 
For the analysis of $P_i$'s dilation, we first consider an $(\alpha, \alpha + 1)$-ruling 
set of $\hat{P}_i$ for $\alpha = 12\kappa_D \log^3 n$, which is denoted 
by $S = \{s_0, s_1, \dots, s_z\}$. Note that this ruling set is introduced only for the analysis, and the algorithm does not construct it actually. The key 
observation of the proof is that for any $s_j$ ($1 \leq j \leq z$) $H_i$ 
contains a path 
of length $\tilde{O}(\kappa_D)$ from $s_0$ to $s_j$ whp. 
It follows that any two nodes $u, v \in V_{\hat{P}_i}$ are connected by a path of 
length $\tilde{O}(\kappa_D)$ in $P_i + H_i$ because any node in 
$V_{\hat{P}_i}$ has at least one ruling-set node within distance $\alpha + 1$ 
in $P_i + H^1_i$.

To prove the claim above, we further introduce the notion 
of \emph{terminal sets}.
A terminal set $T_j \subseteq V_{P_i}$ associated with $s_j \in S$ ($0 \leq j \leq z$) is the subset of $V_{P_i}$ satisfying (1) $|T_j| \geq 
\kappa_D \log^3 n$, (2) $\Dist_{P_i + H_i}(s_j, x) \leq 6 \kappa_D \log^3 n$ for any
$x \in T_j$, and (3) $N^+(x) \cap N^+(y) = \emptyset$ for any $x, y \in T_j$ (notice that $N^+(\cdot)$ is the set of neighbors in $G$, not in $P_i + H^1_i$).
We can show that such a set always exists. 

\begin{lemma}\label{lma:terminal}
Letting $S = \{s_0, s_1, \dots, s_z\}$ be any $(\alpha, \alpha + 1)$-ruling 
set of $\hat{P}_i$ for $\alpha = 14\kappa_D \log^3 n$, there always exists a terminal
set $\mathcal{T} = \{T_0, T_1, \dots, T_z\}$ associated with $S$. 
\end{lemma}
\begin{proof}
The proof is constructive. Let $c = 6 \kappa_D \log^3 n$ for short. 
We take an arbitrary shortest path $Q = (s_j = u_0, u_1, u_2, \dots, 
u_{c})$ of length $c$ in $P_i + H^1_i$ starting from $s_j \in S$.  
Since no two nodes in $N^+(V_{P_i}) \setminus V_{P_i}$ are adjacent
in $P_i + H^1_i$, $Q$ contains no two consecutive nodes which are both
in $N^+(V_{P_i}) \setminus V_{P_i}$. It implies that at least
half of the nodes in $Q$ belongs to $V_{P_i}$. Let $q' = (u'_{0}, u'_{1}, 
\dots u'_{c'})$ be the subsequence of $Q$ consisting of the nodes in $V_{P_i}$.
Then we define $T_j = \{u'_{0}, u'_{3}, \dots, 
u'_{3\lfloor c'/3 \rfloor}\}$, which satisfies the three properties of 
terminal sets: 
It is easy to check that the first and second properties hold. In addition, 
one can show that $\Dist_G(u'_x, u'_{x+a}) \geq 3$ (which is equivalent to 
$N^+(u'_x) \cap N^+(u'_{x + a}) = \emptyset$) holds for any $a \geq 3$ and 
$x \in [1, c' - a]$: Suppose for contradiction that $\Dist_G(u'_x, u'_{x+a}) \leq 2$ holds for some $a \geq 3$ and $x \in [1, c' - a]$. The distance two 
between $u'_x$ and $u'_{x+a}$ implies $N^+(u'_x) \cap N^+(u'_{x+a}) \neq \emptyset$, and thus $\Dist_{\hat{P}_i}(u'_x, u'_{x+a}) \leq 2$ holds. Then 
bypassing the subpath from $u'_{x}$ to $u'_{x+a}$ in $Q$ through the distance-two
path we obtain a path from $s_j$ to $u_c$ shorter than $Q$. It contradicts the fact that $Q$ is the shortest path. 
\end{proof}

The second property of terminal sets and the following lemma deduces 
the fact that $\Dist_{P_i + H_i}(s_0, s_j) = \tilde{O}(\kappa_D)$ holds 
for any $j \in [0, z]$.

\begin{lemma}
\label{lma:dilation_small}
Letting $S = \{s_0, s_1, \dots, s_z\}$ be any $(\alpha, \alpha + 1)$-ruling 
set of $\hat{P}_i$ for $\alpha = 14\kappa_D \log^3 n$, and 
$\mathcal{T} = \{T_0, T_1, \dots, T_z\}$ be a terminal set associated with
$S$. For any $j \in [0, z]$, there exist $u \in T_0$ and $v \in T_j$ such that 
$\Dist_{P_i + H_i}(u, v) = O(1)$ holds.
\end{lemma}
\begin{proof}
Since the distance of $s_0$ and $s_j$ is at least $14\kappa_D\log^3 n$, we have $N^+(T_0) \cap N^+(T_j) = \emptyset$.
The proof is divided into the cases of $D = 3$ and $D = 4$.
($D = 3$) By the conditions of $N^+(T_0) \cap N^+(T_j) = \emptyset$ and $D = 3$,
there exists a path of length exactly three from any node $a \in T_0$ to any node $b \in T_j$. Letting $e_{a, b}$ be the second edge in that path, we define 
$F = \{e_{a, b} \mid a \in T_0, b \in T_j\}$. By the third property of terminal sets and the fact of $N^+(T_0) \cap N^+(T_j) = \emptyset$,
for any two edges $(x_1, y_1), (x_2, y_2) \in F$, either $x_1 \neq x_2$ or 
$y_1 \neq y_2$ holds. That is, $e_{a_1, b_1} \neq e_{a_2, b_2}$ holds for
any $a_1, a_2 \in T_0$ and $b_1, b_2 \in T_j$. By the second property of terminal sets, it implies $|F| = |T_0||T_j| 
\geq (\kappa_D \log^3 n)^2$. Since each edge in $F$ is added to $H^2_i$ 
with probability
$1/n^{1/2} = 1/ \kappa_D^2$, the probability that no edge in $F$ is added to $H^2_i$ is at most $(1 - 1 / \kappa_D^2)^{(\kappa_D \log^3 n)^2} \leq e^{-\Omega(\log^6 n)}$. That is, an edge $e_{a, b}$ is added to $H_i$ whp. and then
$\Dist_{P_i + H_i}(a, b) \leq 3$ holds. 
($D = 4$) For any node $u \in T_0$ and $v \in T_j$, there exists a path
from $u$ to $v$ of length three or four in $G$. That path necessarily contains
a length-two sub-path $P_2(u, v) = (a_{uv}, b_{uv}, c_{uv})$ such that 
$a_{uv} \in N^+(u)$ and $c_{uv} \in N^+(v)$ holds (if $P_2(u, v)$ is not uniquely
determined, an arbitrary one is chosen). We call $(a_{uv}, b_{uv})$ and 
$(b_{uv}, c_{uv})$ the \emph{first} and \emph{second edges} of $P_2(u, v)$
respectively. Let $\mathcal{P}_2 = \{P_2(u, v) \mid 
u \in T_0, v \in T_j\}$, $G'$ be the union of $P_2(u, v)$ for all 
$u \in T_0$ and $v \in T_j$, and $\mathcal{P}^e_2 = \{P_2(u, v) \in \mathcal{P}_2 
\mid e \in P_2(u, v)\}$ for any $e \in E_{G'}$. 
We first bound the size of $\mathcal{P}^e_2$. Assume that 
$e$ is a first edge of some path in $\mathcal{P}^e_2$. Let $e = (a, b)$ and 
$u \in T_0$ be the (unique) node such that $a \in N^+(u)$ holds. 
Since at most $|T_j|$ paths in $\mathcal{P}_2$ can start from a node in 
$N^+(u)$, the number of paths in $\mathcal{P}_2$ using $e$ as their first edges
is at most $|T_j|$. Similarly, if $e$ is the second edge of some path in $\mathcal{P}^e_2$, at most $|T_0|$ paths in $\mathcal{P}_2$ can contain $e$ as their 
second edges. While some edge may be used as both first and second edges, the total number of paths using $e$ is bounded by $|T_0| + |T_j| = 2 \kappa_D \log^3 n$. It implies that any path $P_2(u, v)$ can share edges with at most 
$4 \kappa_D \log^3 n$ edges, and thus $\mathcal{P}_2$ contains at least 
$|T_0||T_j| / (4\kappa_D \log^3 n + 1) \geq \kappa_D \log^3 n/ 5$ 
edge-disjoint paths. Let $\mathcal{P}'_2 \subseteq \mathcal{P}_2$ be the maximum-cardinality subset of $\mathcal{P}_2$ such that any $P_2(u_1, v_1), 
P_2(u_2, v_2) \in \mathcal{P}'_2$ is edge-disjoint. We define 
$B = \{ b \mid (a, b, c) \in \mathcal{P}'_2\}$. Let $\Delta(b)$ be the number of
paths in $\mathcal{P}'_2$ containing $b \in B$ as the center. Due to the 
edge disjointness of $\mathcal{P}'_2$, we have $|E_{G}(N^+(T_0), b)| \geq 
\Delta(b)$ and $|E_{G}(N^+(T_j), b)| \geq \Delta(b)$ for any $b \in B$. 
Let $Y_b$ be the value of $h(b, i)$, and $X_b$ 
be the indicator random variable that takes one if a path in $\mathcal{P}'_2$ which contains $b$ as the center is added to $H_i$, and zero otherwise. 
Let $X$ and $Y$ be the indicator random variables corresponding to the 
events of $\bigvee_{b \in B} X_b = 1$ and $\bigvee_{b \in B} Y_b \leq 
\Delta(b)/\log^2 n$ respectively.
Then we obtain $\Pr[X_b = 1 \mid Y_b = y] \geq 1 - \left(1 - 1/y\right)^{\Delta(b)} \geq 1- 2e^{-\Delta(b)/y},$
 and thus $\Pr[X_b = 1 \mid Y_b \leq \Delta(b)/\log^2 n] \geq  
1 - e^{-\Omega(\log^2 n)}$ holds. That is, $\Pr[X = 1 \mid Y = 1] 
\geq 1 - e^{-\Omega(\log^2 n)}$ holds. Since $h$ is $(n^{1/3}\log^3 n)$-wise
independent, $Y_b$ for all $b \in B$ are independent. Thus we obtain 
\begin{align*}
\Pr[Y = 1] &= 1 - \Pr[Y = 0] \\
&= 1 - \Pr\left[\bigwedge_{b \in B} Y_b> \frac{\Delta(b)}{\log^2 n}\right] \\
&= 1 - \prod_{b \in B} \Pr\left[Y_b > \frac{\Delta(b)}{\log^2 n}\right] \\
&= 1 - \prod_{b \in B} \left(1 - \frac{\Delta(b)}{n^{\frac{1}{3}} \log n}\right) \\
&\geq 1 - e^{-\sum_{b \in B} \frac{\Delta(b)}{n^{\frac{1}{3}} \log n}} \\
&= 1 - e^{-\frac{|\mathcal{P}'_2|}{n^{\frac{1}{3}} \log n}}\\
&\geq 1 - e^{-\Omega(\log^2 n)}.
\end{align*}
Consequently, we have $\Pr[X = 1] \geq \Pr [X = 1 \wedge Y = 1] \Pr[Y = 1]  \geq \left(1 - e^{-\Omega(\log n)}\right)^2.$
The lemma is proved.
\end{proof}

\subsection{Distributed Implementation}

We explain below the implementation details of the algorithm stated above in 
the CONGEST model.
\begin{itemize}
    \item (\textbf{Preprocessing}) In the algorithm stated above, the shortcut construction is performed only for large parts, which is crucial to bound the congestion of each edge. Thus, as a preprocessing task, each node has to 
    know if its own part is large (i.e. having a diameter larger than $\kappa_D$) or not. While the exact identification of the diameter 
    is usually a hard task, just an asymptotic identification is 
    sufficient for achieving the shortcut quality stated above, where 
    the parts of diameter $\omega(\kappa_D)$ and diameter $o(\kappa_D)$ must be
    identified as large and small ones, but those of diameter 
    $\Theta(\kappa_D)$ is identified arbitrarily. This loose identification is easily implemented by a simple distance-bounded aggregation. The algorithm 
    for part $P_i$ is that: (1)At the first round, each node in $P_i$ sends its ID to all the neighbors, and (2)in the following rounds, each node 
    forwards the minimum ID it received so far. The algorithm executes this message propagation during $\kappa_D$ rounds. If the diameter is (substantially) larger than $\kappa_D$, the minimum ID in $P_i$ does not 
    reach all the nodes in $P_i$. Then there exists an edge whose endpoints
    identify different minimum IDs. The one-more-round propagation allows
    those endpoints to know the part is large. Then they start to 
    broadcast the signal ``large'' using the following $\kappa_D$ rounds. 
    If $\kappa_D$ is large, the signal ``large'' is invoked at several nodes in $P_i$, and $\kappa_D$-round propagation guarantees that every node receives the signal. That is, any node in $P_i$ identifies that
    $P_i$ is large. The running time of this task is $O(\kappa_D)$ rounds.      
    \item (\textbf{Step 1}) As we stated, the 1-hop extension is implemented 
    in one round. In this step, each node $v \in V_{P_i}$ tells all the 
    neighbors if $P_i$ is large or not. Consequently, if part $P_i$ 
    is identified as a large one, all the nodes in $N^+(P_i)$ know it 
    after this step. 
    \item (\textbf{Step 2}) The algorithm for $D=3$ is trivial. For $D = 4$,
    there are two non-trivial matters. The first one is the preparation of
    hash function $h$. We realize it by sharing a random seed of 
    $O(n^{1/3}\log^3 n\log |\mathcal{Y}|)$-bit length in advance. 
    A standard construction by Wegman and Carter~\cite{WC81} allows  
    each node to construct the desired $h$ in common. Sharing the 
    random seed is implemented by the broadcast of one $O(n^{1/3}\log^3 n
    \log |\mathcal{Y}|)$-bit message, i.e., taking $\tilde{O}(\kappa_D)$ rounds. The second matter is to address the fact that $u$ does not know
    if $P_i$ is large or not, and/or if $v$ belongs to $N^+(P_i)$ or not. It makes
    $u$ difficult to determine if $(u, v)$ should be added to $H_i$ or not. 
    Instead, our algorithm simulates the task of $u$ by the nodes in 
    $N(u)$. More precisely, each node $v \in N^+(V_{P_i})$ adds each 
    incident edge $(u, v)$ to $H_i$ with probability $1/h(u, i)$. 
    Due to the fact of $v \in N^+(P_i)$, $v$ knows if $P_i$ is large or not
    (informed in step 1), and also can compute $h(u, i)$ locally. 
    Thus the choice of $(u, v)$ is locally decidable at $v$.
    Since this simulation is completely equivalent to the centralized 
    version, the analysis of the quality also applies.
\end{itemize}

It is easy to check that the construction time of the distributed implementation
above is $\tilde{O}(\kappa_D)$ in total.

\section{Low-Congstion Shortcut for Bounded Clique-width Graphs}
\label{sec:clique-width}
\label{sec:def_clique}
Let $G = (V, E)$ a graph. A $k$-graph ($k \geq 1$) is a graph whose vertices are labeled by integers in $[1, k]$. A $k$-graph is naturally defined as a triple
$(V, E, f)$, where $f$ is the labeling function $f : V \to [1, k]$. The 
clique-width of $G= (V, E)$ is the minimum $k$ such that there exists a
$k$-graph $G = (V, E, f)$ which is constructed by means of repeated application 
of the following
four operations: (1) introduce: create a graph of a single node $v$ with label 
$i \in [1,k]$, (2) disjoint union: take the union $G \cup H$ of two $k$-graphs $G$ and 
$H$, (3) relabel: given $i, j \in [1, k]$, change all the labels 
$i$ in the graph to $j$, and (4) join: given $i, j \in [1, k]$, 
connect all vertices labeled by $i$ with all vertices labeled by $j$ by edges.   

The clique-width is invented first as a parameter to capture the tractability for 
an easy subclass of high treewidth graphs~\cite{courcelle2000upper,corneil2005relationship}. 
That is, the class of bounded clique-width can contain many graphs with high treewidth. 
In centralized
settings, one can often obtain polynomial-time algorithms for many 
NP-complete problems under the assumption of bounded clique-width. 
The following negative result, however, states that bounding clique-width 
does not admit any good solution
for the MST problem (and thus also for the low-congestion shortcut). 
\begin{theorem}
\label{theo:clique}
There exists an unweighted $n$-vertex
graph $G = (V, E)$ of clique-width six where for any MST algorithm $A$ 
there exists an edge-weight function $w_A : E \to \mathbb{N}$ such that the 
running time of $A$ becomes $\tilde{\Omega}(\sqrt{n}+D)$ rounds.
\end{theorem}

We introduce the instance stated in this theorem, 
which is denoted by $G(\Gamma, p)$ ($\Gamma$ and $p$ are the parameters fixed later), using 
the operations specified in the definition of clique-width. That is, this introduction itself becomes
the proof of clique-width six.  
Let $\mathcal{G}(\Gamma)$ be the set of $6$-graphs that contains one node with 
label 1, $\Gamma$ nodes with label 2, and $\Gamma$ nodes label 3, and all other
nodes are labeled by 4. Then we define the binary operation $\oplus$ over 
$\mathcal{G}(\Gamma)$. For any $G, H \in \mathcal{G}(\Gamma)$, the graph 
$G \oplus H$ is defined as the one obtained by the following operations: 
(1) Relabel $2$ in $G$ with $5$ and relabel $3$ in $H$ with $6$, (2)
take the disjoint union $G \cup H$, (3) joins with labels $5$ and $6$, 
(4) relabel $5$ and $6$ with $4$, and then $1$ with $5$, (5) Add  
a node with label $1$ by operation introduce (6) join with 1 and 5, and (7) relabel 5 with 4. This process is 
illustrated in Figure~\ref{fig:construct}.
\begin{figure}[ht]
\centering
\includegraphics[width=100mm,keepaspectratio]{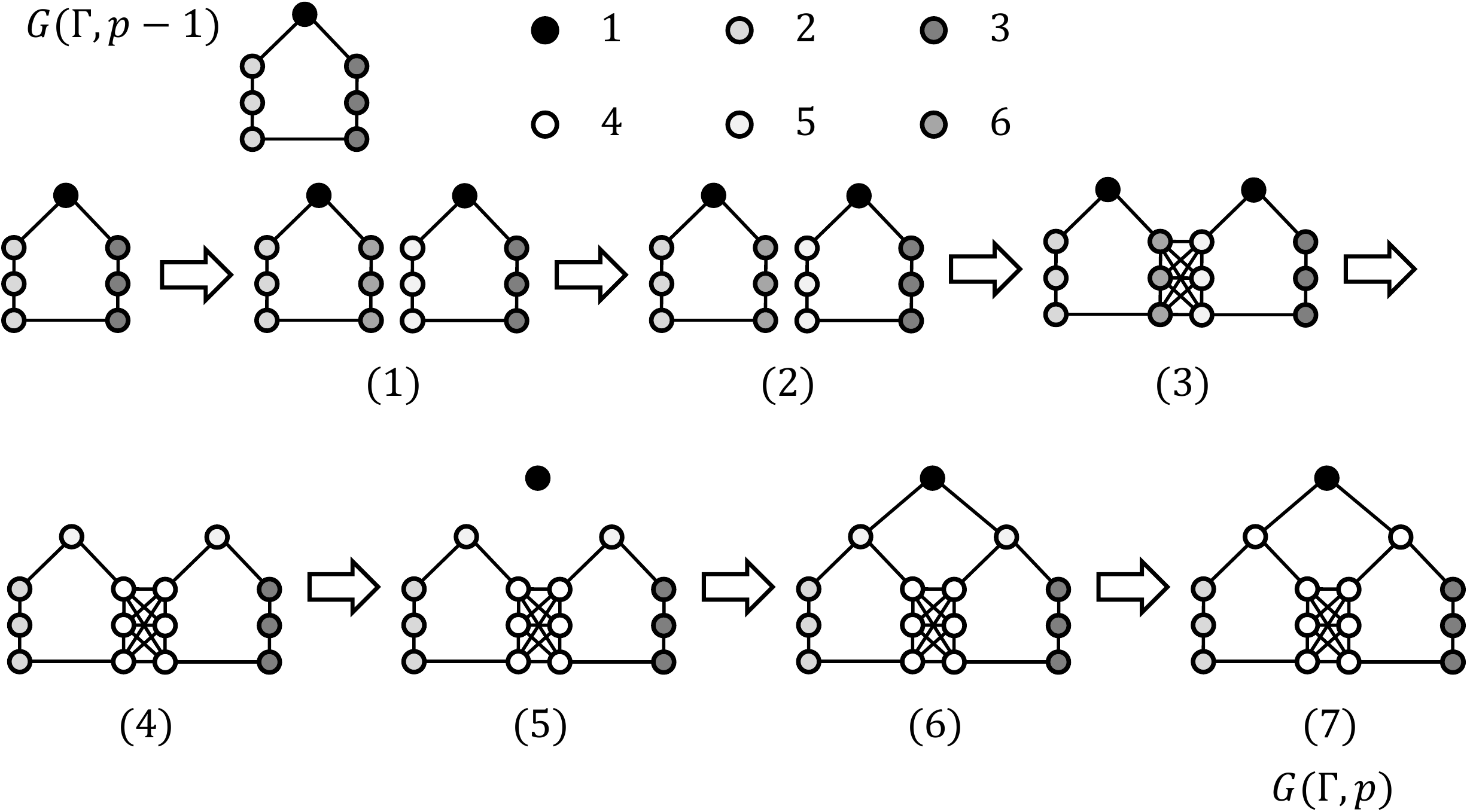}
\caption{Graph $G \oplus H$.}
\label{fig:construct}
\end{figure}

Now we are ready to define $G(\Gamma, p)$. The construction is recursive.
First, we define $G(\Gamma, 1)$ as follows: (1) Prepare a $(2\Gamma)$-biclique 
$K_{\Gamma, \Gamma}$ where one side has label 2, and the other side has label 3.
Note that two labels suffice to construct $K_{\Gamma, \Gamma}$. (2)
Add three nodes with label 1, 5, and 6 by operation introduce. (3) Join with label 2 and 5, 
and with 3 and 6. (4) Join with label 1 and 5, 
and with 1 and 6. (5) Relabel 5 and 6 with 4. Then, we define $G(\Gamma, p) = 
G(\Gamma, p-1) \oplus G(\Gamma, p- 1)$. The instance claimed in 
Theorem~\ref{theo:clique} is $G(\sqrt{n}, \log n /2)$, which is illustrated
in Figure~\ref{fig:clique}. This instance is very close to the standard
hard-core instance used in the prior work (e.g., \cite{PR00,lowerbound}. See 
Figure~\ref{fig:lowerMST}). Thus it is not difficult to see that 
$\tilde{\Omega}(\sqrt{n})$-round lower bound for the MST construction 
also applies to $G(\sqrt{n}, \log n / 2)$. It suffices to show that the following lemma. Combined with Theorem~\ref{thm:communication_complexity}, we obtain Theorem~\ref{theo:clique}.
\begin{lemma}
\label{clique-width_lowerbound}
$G(\Gamma,p) \in \mathcal{G}(O(\Gamma(2^{p}+2)),\Gamma,2^{p}+2,3p)$.
\end{lemma}
    \begin{proof}
    First, let us formally specify the graph $G(\Gamma, p)$, which is defined
    as follows (vertex IDs introduced below are described in Figure~\ref{fig:clique}): 
        \begin{itemize}
            \item $V(\Gamma,p)=T\cup \bigcup_{1\leq l \leq \Gamma} V^{l}$ such that
            $T=\{u^{j}_{i} \mid 0 \leq i \leq 2^{p}-1, 0 \leq j \leq p\}$,
            $V^{l}=\{v^{l}_{i} \mid 0 \leq i \leq 2^{p}-1\}$.
            \item $E(\Gamma,p)=E_{1}\cup E_{2}\cup E_{3}$ such that
            $E_{1}=\{(u^{j}_{i},u^{j-1}_{\lfloor \frac{i}{2} \rfloor}) \mid 0 \leq i \leq 2^{j}-1,1\leq j \leq p\}$,
            $E_{2}=\{(u^{p}_{i},v^{j}_{i})\mid 0 \leq i \leq 2^{p}-1, 1 \leq j \leq \Gamma\}$,
            $E_{3}=\{(v^{j}_{i},v^{k}_{i+1})\mid 0 \leq i \leq 2^{p}-2, 1\leq j \leq \Gamma,1\leq k \leq \Gamma\}$.
            \end{itemize}
            
    We define $\mathcal{X}$ and $\mathcal{Q}$ for graph $G(\Gamma,p)$ 
    as follows:
    \begin{align*}
    \mathcal{X} &= \{X_1,X_2,\dots,X_{2^{p}+2}\} \ \ \text{s.t.} \\ 
    X_{i}&=
    \begin{cases}
            \left\{u^{p}_{0}\right\}& \hspace*{12mm} \text{$\left(i=1\right)$}\\
            \left\{v^{j}_{0}\mid 1\leq j \leq \Gamma\right\}& \hspace*{12mm} \text{$\left(i=2\right)$}\\
            \left\{v^{j}_{i-2}\mid 2\leq j \leq N\right\}\cup \left\{u^{p-j}_{\frac{i-1}{2^{j}}-1}\mid 0\leq j \leq p, i-1\bmod{2^{j}}=0  \right\}& \hspace*{12mm} \text{$\left(3\leq i \leq 2^{p}-1 \right)$}\\
            \left\{v^{j}_{2^{p}-2}\mid 1\leq j \leq \Gamma\right\}\cup \left\{u^{p}_{2^{p}-2}\right\}\cup \left\{u^{j}_{2^{j}-1} \mid 0\leq j \leq p-1 \right\}& \hspace*{12mm} \text{$\left(i=2^{p}\right)$}\\
            \left\{v^{j}_{2^{p}-1}\mid 1\leq j \leq \Gamma\right\}& \hspace*{12mm} \text{$\left(i=2^{p}+1\right)$}\\
            \left\{u^{p}_{2^{p}-1}\right\}& \hspace*{12mm} \text{$\left(i=2^{p}+2\right)$}.
    \end{cases}\\
    \end{align*}
    \begin{align*}
    \mathcal{Q} &= \{Q_1,Q_2,\dots,Q_{\Gamma}\} \ \ \text{s.t.} \\ 
    Q_{i}&=
    \begin{cases}
            V_{1}\cup \left(T\backslash (s \cup r)\right)& \hspace*{12mm} \text{$\left(i=1\right)$}\\
            V_{i}& \hspace*{12mm} \text{$\left(2\leq i\leq \Gamma\right)$}.
    \end{cases}\\
    \end{align*}
    It is easy to check (\textbf{C1}) and (\textbf{C2}) is satisfied.
    Thus we only show that (\textbf{C3}) is satisfied.
    Let $V_{R_{i}}=R_{i} \cap \bigcup_{j=1}^{\Gamma}V_j$.
    For $2\leq i\leq (2^{p}+2)/2$, we have $(N(V_{R_{i}})\backslash R_{i-1}) = \emptyset$.
    For any $\ell$ and $1\leq i \leq 2^{p-2}$, if $u^{p}_{i}$ is included in $R_{\ell}$, then the neighbors of $u^{p}_{i}$ is included in $R_{\ell}$.
    For any $\ell$, $1\leq i \leq p$ and  $0\leq j \leq 2^{i}-2$, if $u^{i}_{j}$ is included in $R_{\ell}$, then $u^{i}_{j+1}$ is included in $R_{\ell}$.
    Let $u^{i}(R_{\ell})$ be leftmost vertex which level is $i$ of $T$ and included in $R_{\ell}$.
    For any $\ell$, $1\leq i \leq p$ and $0\leq j \leq 2^{i}-1$, if $u^{i}_{j}\neq u^{i}(R_{\ell})$ and $u^{i}_{j}$ is included in $R_{\ell}$, then the parent of $u^{i}_{j}$ is included in $R_{\ell}$.
    Thus $|(N(R_\ell)\backslash R_{\ell-1})|$ only includes neighbors of $u^{i}(R_{\ell})$ for $1\leq i \leq p$ and $2\leq \ell\leq (2^{p}+2)/2$.
    Since the tree $T$ is binary tree, $u^{i}(R_{\ell})$ has at most 3 neighbors in $T$.
    Therefore we have $|E\left((N(R_i)\backslash R_{i-1}) \right)|\leq 3p$.
    Similarly, we have $|E\left((N(L_i)\backslash L_{i-1}) \right)|\leq 3p$.
    Therefore we can prove that the graph $G(\Gamma,p)$ is included in $\mathcal{G}(O(\Gamma(2^{p}+2)),\Gamma,2^{p}+2,3p)$.
    By Theorem~\ref{thm:communication_complexity}, the lower bound of constructing MST in $\mathcal{G}(O(\Gamma(2^{p}+2)),\Gamma,2^{p}+2,3p)$ is $\tilde{\Omega}((\min\{\Gamma/3p,\left((2^{p}+2\right)/2-1\})$. When $\Gamma=\Theta(\sqrt{n})$ and $2^{p}=\Theta(\sqrt{n})$, we obtain the $\tilde{\Omega}(\sqrt{n})$ lower bound.
\end{proof}

\begin{figure}[ht]
\centering
\includegraphics[width=100mm,keepaspectratio]{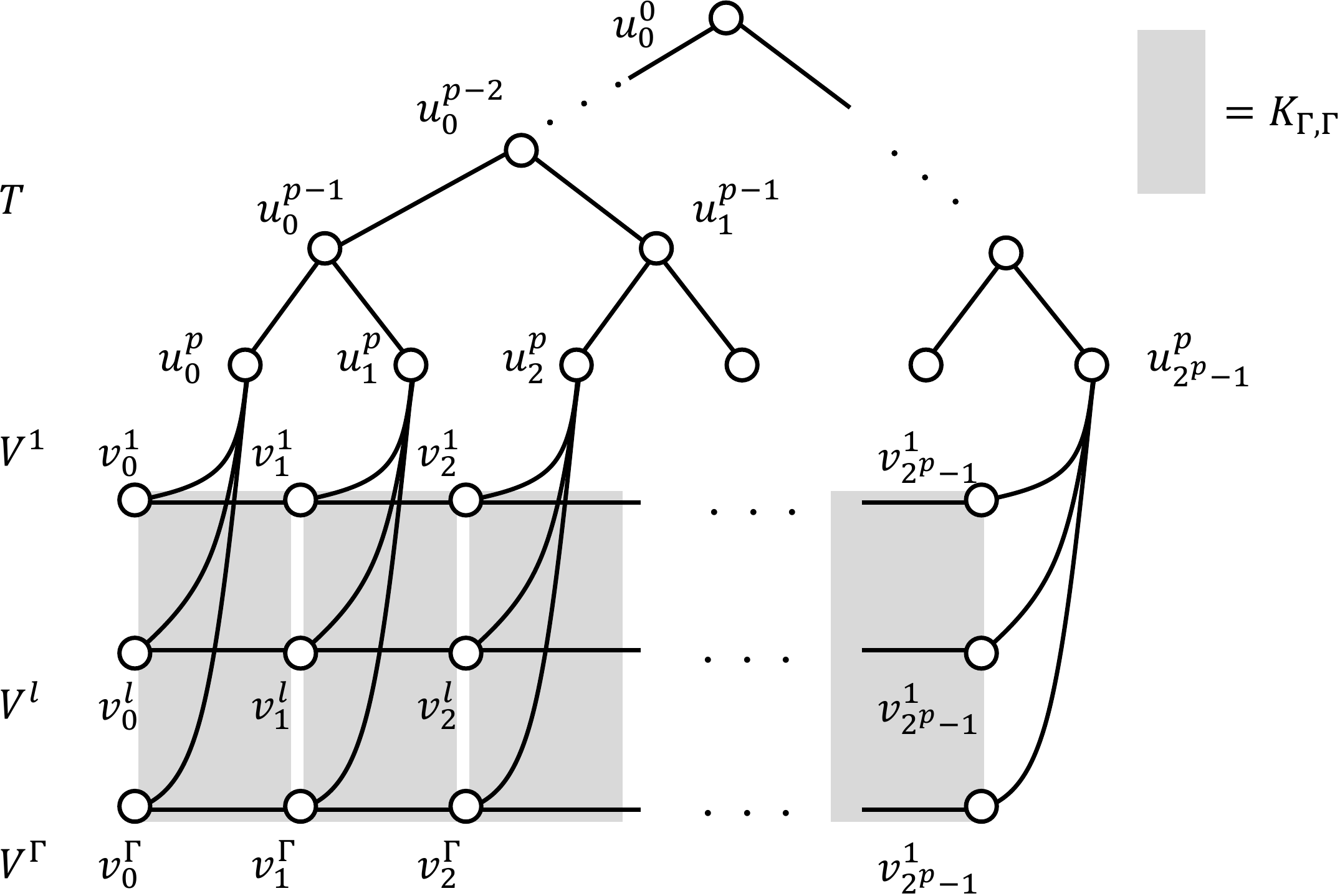}
\caption{Example of clique-width 6 graph $G(\Gamma,p)$.}
\label{fig:clique}
\end{figure}
%

%
%
%
%
%
%
%
%

\section{Conclusion}
\label{sec:conclusion}
In this paper, we have shown the upper and lower bounds for the
round complexity of shortcut
construction and MST in $k$-chordal graphs, diameter-three or four graphs, 
and bounded clique-width graphs. We presented an $O(1)$-round algorithm
constructing an optimal $O(kD)$-quality shortcut for any $k$-chordal graphs. 
We also presented the algorithms of constructing optimal low-congestion shortcuts with quality $\tilde{O}(\kappa_D)$ in $\tilde{O}(\kappa_D)$ rounds for $D=3$ 
and $4$, which yield the optimal algorithms for MST matching the known lower 
bounds by Lotker et al.~\cite{LPP06}. On the negative side, $O(1)$-clique-width
does not allow us to have good shortcuts. 
We conclude this paper posing three related open problems. (1) Can we have 
good shortcuts for $D \geq 5$? (2) Can we have 
good shortcuts for $k$-clique width where $k\leq 5$? (3) While bounded clique-width does not contribute to solving MST efficiently, it seems to provide many edge-disjoint paths (not necessarily so short). Can we find any problem that can uses the benefit of bounded clique-width?
%
%
\section*{Acknowledgements}
This work was supported by JSPS KAKENHI Grant Numbers JP18H04091, JP18K11168, JP18K11169, JP19K11824, and JP19J22696, and JST SICORP Grant Number JPMJSC1606, Japan.
%
%

\bibliography{reference}
\end{document}